\documentclass[10pt,conference]{IEEEtran}
\usepackage{amsmath,epsfig,amssymb,amsthm}
\usepackage{stmaryrd}
\usepackage[numbers, square, comma, compress]{natbib}
\usepackage{dsfont}
\usepackage{graphicx}
\usepackage{epstopdf}

\newcommand{\F}{\mathbb{F}}
\newcommand{\Q}{\mathbb{Q}}
\newcommand{\C}{\mathbb{C}}

\newcommand{\R}{\mathbb{R}}

\newcommand{\ZZ}{\mathbb{Z}}
\newcommand{\OO}{\mathcal{O}}

\providecommand{\abs}[1]{\ensuremath{\left\lvert #1 \right\rvert}}

\providecommand{\norm}[1]{\ensuremath{\left\Vert #1 \right\Vert}}

\providecommand{\vv}[1]{\textquotedblleft #1\textquotedblright}

\DeclareMathOperator*{\Vol}{Vol}

\newtheorem{lem}{Lemma}
\theoremstyle{definition}
\newtheorem{rem}{Remark}

\newtheorem{prop}{Proposition}
\newtheorem{defi}{Definition}
\newtheorem{theorem}{Theorem}
\newtheorem{cor}{Corollary}

\begin{document}
\title{Number field lattices achieve Gaussian and Rayleigh channel capacity within a constant gap}
\author{ 
\IEEEauthorblockN{Roope Vehkalahti}
\IEEEauthorblockA{Department of Mathematics and Statistics,
University of Turku\\
Finland\\
roiive@utu.fi}
\and
\IEEEauthorblockN{Laura Luzzi}
\IEEEauthorblockA{Laboratoire ETIS
 (ENSEA - UCP - CNRS) \\
Cergy-Pontoise, France \\
laura.luzzi@ensea.fr}

}

\maketitle
\begin{abstract}
This paper shows that a family of number field lattice codes simultaneously achieves a constant gap to capacity in Rayleigh fast fading and Gaussian channels. 
The key property in the proof is the existence of infinite towers of Hilbert class fields with bounded root discriminant. The gap to capacity of the proposed lattice codes is determined by the root discriminant. \\ 
The comparison between the Gaussian and fading case reveals that in Rayleigh fading channels the \emph{normalized minimum product distance} plays an analogous role to the Hermite invariant in Gaussian channels.
\end{abstract}
\section{Introduction}
The classical problem of achieving the capacity of the Gaussian channel using structured codes has seen significant recent advances.
In particular, random lattice code ensembles have been shown to attain capacity \cite{DeBuda, Urbanke_Rimoldi}. Good lattice code ensembles can be constructed by lifting linear codes over finite fields \cite{Loeliger,Erez_Zamir} or using multilevel codes \cite{Forney_Trott_Chung}; an explicit multilevel construction from polar codes was recently proposed in \cite{Yan_Ling_Wu}. \\   
In this paper, we consider an alternative approach based on algebraic number theory. It is well-known that lattice constellations from number fields provide good performance on Gaussian and fading channels \cite{BERB, FOV}. As far as we know, the problem of achieving ergodic capacity with structured codes is still open in the case of fading channels. \\
In this work, we analyze the asymptotic behavior of algebraic lattices from number fields when the lattice dimension tends to infinity, and  
show that Hilbert class field towers with bounded root discriminants simultaneously reach a constant gap to capacity on both Gaussian and Rayleigh fading channels.  \\
We note that the constant gap to capacity is achieved not only using ML decoding, but also with simple naive lattice decoding. \\
While we discuss specific number field lattices, our proofs do work for any ensemble of lattices with asymptotically good product distance.  The larger the product distance, the smaller the gap to the capacity in the fast fading channel. 

In the existing literature, the product distance is  mostly seen as a rough tool to estimate the worst case pairwise error probability in the  high SNR regime. Instead we will see that when we are allowed to decode and encode over a growing number of time units the normalized product distance will play a role of  an equal importance to the Hermite constant in Gaussian channels. We point out that the study of normalized product distance and Hermite invariant are both examples of  
the more general problem of finding the minima of homogeneous forms in the  mathematical field of \emph{geometry of numbers}. This seems to be a  universal theme, where each fading channel model is linked to a natural problem in geometry of numbers. We will elaborate further on this topic in \cite{LV2014}, where we also extend our capacity results to the MIMO context.\\
The families of number fields we consider were first brought to coding theory in \cite{LT}, where the authors pointed out that the corresponding lattices  have large Hermite constant. Our proof for the Gaussian channel is therefore an obvious corollary to this result. In \cite{Xing} it was pointed out that these families of number fields provide the best known normalized product distance.


\section{Notation and preliminaries}

In this section we will use the notation $\F$ for the field $\R$ or $\C$.\\
A   {\em lattice} $L \subset \F^n$ has the form $L=\ZZ x_1\oplus \ZZ x_2\oplus \cdots \oplus \ZZ x_{k}$, where the vectors $x_1,\dots, x_{k}$ are linearly independent over $\R$, i.e., form a lattice basis.

\begin{defi}\label{euclidean norm}
Let $v=(v_1, ...,v_n)$ be a vector in $\F^n$.
The \emph{Euclidean norm}   of $v$ is $||v||_E=\sqrt{\sum_{i=1}^n |v_i|^2}$. If $L$ is a   lattice in $\F^n$,  the \emph{minimum distance} $\mathrm{sv}(L)$ of $L$
is defined to be the infimum of the  Euclidean  norms of   all non-zero vectors in the lattice.
\end{defi}

\begin{defi}\label{tulo}
Let $v=(v_1,...,v_n)$ be a vector in $\F^n$.
We define  the \emph{ product norm} of $v$ as
$
n(v)=\prod_{i=1}^n |v_i|.
$

Assuming that $n(v)\neq 0$ for all the non zero elements $v \in L$, we can define the \emph{minimum product distance} $\mathrm{d}_{\mathrm{p, min}}(L)$ of $L$
to be the infimum of the product norms of all non-zero vectors in the lattice.
\end{defi}
We will use the notation $\Vol(L)$ for the volume of the fundamental parallelotope of the lattice $L$.

We  denote by $\mathrm{Nd}_{\mathrm{p, min}}(L)$ the {\em normalized minimum product distance} of the lattice $L$, i.e. here we first scale $L$ to have a unit size fundamental parallelotope and then take $\mathrm{d}_{\mathrm{p, min}}(L')$ of the resulting lattice $L'$. In the same way we can define the normalized
shortest vector of $L$ and denote it with $\mathrm{Nsv}(L)$. The square of the normalized shortest vector is called the \emph{Hermite invariant} of the lattice.

We then have the following scaling laws.
If $L$ is a full lattice in $\C^n$, then
$$
\mathrm{Nd}_{\mathrm{p, min}}(L)=\frac{\mathrm{d}_{\mathrm{p, min}}(L)}{\Vol(L)^{1/2}}, \quad \mathrm{Nsv}(L)=\frac{\mathrm{sv}(L)}{\Vol(L)^{1/2n}}.
$$
In the case of a real lattice $L\subset \R^n$ we have 
$$
\mathrm{Nd}_{\mathrm{p, min}}(L)=\frac{\mathrm{d}_{\mathrm{p, min}}(L)}{\Vol(L)},\quad \mathrm{Nsv}(L)=\frac{\mathrm{sv}(L)}{\Vol(L)^{1/n}}.
$$

These two concepts are related by the following simple and well known application of the arithmetic-geometric mean inequality.
\begin{prop}\label{simple}
Let $L$ be a lattice in $\F^n$. Then
$$
\mathrm{Nd}_{\mathrm{p, min}}(L)\leq \frac{\mathrm{Nsv}(\phi(L))^n}{n^{n/2}}.
$$
\end{prop}

The following Lemma \cite{GL} is useful in order to choose lattice constellations with prescribed minimum size. 

\begin{lem} \label{shift}
Let us suppose that $L$ is a full lattice in  $\F^n$  and  $S$  a Jordan measurable bounded subset of $\F^n$. Then there exists $x\in \F^n$ such that
$$
|(L+x)\cap S|\geq\frac{\mathrm{Vol}(S)}{\mathrm{Vol}(L)}.
$$
\end{lem}

\section{Lattice codes from number fields}
 In the following we will will describe the standard method  to build lattice codes from number fields \cite{BERB}.
We will denote the discriminant of a number field $K$ with $d_K$. For every number field it is a non-zero integer.
\subsection{Complex constellations}\label{complex_sec}
Let $K/\Q$ be  a totally complex extension of degree $2n$ and $\{\sigma_1,\dots,\sigma_n\}$ be a set  of  $\Q$-embeddings, such that we have chosen one from each complex conjugate pair. Then we can define a
\emph{relative canonical embedding} of $K$ into $\C^n$ by
$$
\psi(x)=(\sigma_1(x),\dots, \sigma_n(x)).
$$
The ring of algebraic integers $\OO_K$ has a  $\ZZ$-basis $W=\{w_1,\dots ,w_{2n}\}$ and $\psi(W)$ is a $\ZZ$-basis for the full  lattice $\psi(\OO_K)$ in $\C^n$.

\begin{lem}\label{complex}
Let  $K/\Q$ be an extension of degree $2n$ and  let $\psi$ be the relative canonical embedding. Then
$$
\mathrm{Vol}(\psi(\mathcal{O}_K))=2^{-n}\sqrt{|d_K|} 
$$
$$
\mathrm{{Nd}_{p,min}}(\psi(\mathcal{O}_K))=\frac{2^{\frac{n}{2}}}{|d_K|^{\frac{1}{4}}} \,\, \mathrm{and}\,\,   \mathrm{Nsv}(\psi(\mathcal{O}_K)=\frac{\sqrt{2n}}{|d_K|^{1/4n}}.
$$
\end{lem}
We can now see  that both the normalized product distance and Hermite invariant of the number field lattices depend only on the discriminant of the field.
In order to find promising codes we need fields with as small discriminants as possible.

Martinet \cite{Martinet} proves the existence of an infinite tower of totally complex number fields $\{K_n\}$ of degree $2n$, where $2n=5\cdot2^k$, such that
\begin{equation} \label{G}
 \abs{d_{K_n}}^{\frac{1}{n}}=G^2,
\end{equation}
for $G \approx 92.368$.
For such fields $K_n$ we have that
$$
\mathrm{{Nd}_{p,min}}(\psi(\mathcal{O}_{K_n}))=\left(\frac{2}{G}\right)^{\frac{n}{2}} \,\, \mathrm{and}\,\,\,   \mathrm{Nsv}(\psi(\mathcal{O}_{K_n}))=\frac{\sqrt{2n}}{\sqrt{G}}.
$$
Given transmission power $P$, we require that every point $\mathbf{s}$ in a finite code $\mathcal{C}\subset \C^n$ satisfies the average power constraint
\begin{equation} \label{power_constraint}
\frac{1}{n} \sum_{i=1}^n \abs{s_i}^2=\frac{1}{n} \sum_{i=1}^n (\Re(s_i)^2+\Im(s_i)^2) \leq P.
\end{equation}
Let $R$ denote the code rate in bits per complex channel use; equivalently, $\abs{\mathcal{C}}=2^{Rn}$.
Let us now show how we can produce  codes  $\mathcal{C}$, having rate greater or equal to  $R$, and satisfying the power constraint (\ref{power_constraint}), from the number field lattices  $\psi(\OO_K)$, where $K$ belongs to the Martinet family.

In the following we will use the notation $B(\sqrt{nP})$ for    a $2n$-dimensional ball of radius $\sqrt{nP}$ in  $\C^n$.
Let us suppose that $\alpha$ is some energy normalization constant. According to  Lemma \ref{shift}, we can choose an element $x_R\in \C^n$ such that for $\mathcal{C}=B(\sqrt{nP} )\cap( x_R+\alpha \psi(\mathcal{O}_K))$ we have
$$\abs{\mathcal{C}} \geq 2^{Rn}=\frac{\Vol(B(\sqrt{nP}))}{\Vol(\alpha\psi(\mathcal{O}_K))}=\frac{2^n C_n P^n}{\alpha^{2n}\sqrt{\abs{d_K}}},$$
 where $C_n=\frac{(\pi n)^n}{n!}$.
We can now see that by using the energy normalization 
$$\alpha^2= \frac{2 P (C_n)^{\frac{1}{n}}}{2^R \abs{d_K}^{\frac{1}{2n}}}=\frac{2 P (C_n)^{\frac{1}{n}}}{2^R G}$$ 
the code $\mathcal{C}$ has rate $R$, or greater, and satisfies the average power constraint.

\subsection{Real constellations}\label{real_sec}

Let us now suppose that we have a degree $n$ totally real extension $K/\Q$ and that $\{\sigma_1,\dots,\sigma_n\}$ are the $\Q$ embeddings of $K$. We define the canonical embedding of $K$ into $\R^n$ by
$$
\psi(x)=(\sigma_1(x),\dots, \sigma_n(x)).
$$
We then have that $\psi(\OO_K)$ is an $n$-dimensional lattice in $\R^n$. 

\begin{lem}\label{real}
Let  $K/\Q$ be a totally real extension of degree $n$ and  let $\psi$ be the  canonical embedding. Then
\begin{align*}
&\mathrm{Vol}(\psi(\mathcal{O}_K))=\sqrt{|d_K|}, \\
& \mathrm{{Nd}_{p,min}}(\psi(\mathcal{O}_K))=\frac{1}{\sqrt{|d_K|}} \,\, \mathrm{and}\,\,   \mathrm{Nsv}(\psi(\mathcal{O}_K))=\frac{\sqrt{n}}{|d_K|^{\frac{1}{2n}}}.
\end{align*}
\end{lem}
In the case of totally real fields \cite{Martinet} proves the existence of a family of fields of degree $n$, where $n=2^k$, such that
\begin{equation} \label{GR}
 \abs{d_{K_n}}^{\frac{1}{n}}=G_1,
\end{equation}
where  $G_1 \approx 1058$. If $K$ is a degree $n$ field from this family, 
\begin{equation}\label{realeq}
\mathrm{{Nd}_{p,min}}(\psi(\mathcal{O}_K))=\frac{1}{G_1^{\frac{n}{2}}} \,\, \mathrm{and}\,\,   
\mathrm{Nsv}(\psi(\mathcal{O}_K))=\frac{\sqrt{n}}{\sqrt{G_1}}.
\end{equation}
As in the case of complex constellations,  we will consider finite codes  $\mathcal{C}=B(\sqrt{nP} )\cap( x_{R}+\alpha \psi(\mathcal{O}_K))$, where
$x_{R}$ is chosen so that
$$\abs{\mathcal{C}}\geq 2^{Rn}=\frac{\Vol(B(\sqrt{nP} ))}{\Vol(\alpha\psi(\mathcal{O}_K))}=\frac{C^{\R}_n P^{n/2}}{\alpha^{n}\sqrt{\abs{d_K}}},$$
and $C^{\R}_n=\frac{(\pi n)^{n/2}}{\Gamma(n/2+1)}$.
We then have that the choice
$$\alpha^2= \frac{P(C^{\R}_n)^{\frac{2}{n}}}{2^{2R} \abs{d_K}^{\frac{1}{n}}}=\frac{P(C^{\R}_n)^{\frac{2}{n}}}{2^{2R}G_1},$$ 
yields a code of rate $R$ satisfying the power constraint $(1/n) \sum_{i=1}^n s_i^2 \leq P$.

\section{Number field codes in the Gaussian channel}
Let us now consider the question of the maximal rates we can achieve with the codes $\mathcal{C}$ of the previous section, when we demand vanishing error probability when $n$ grows to infinity.
\subsection{Complex constellations}

We consider a complex Gaussian channel model 
$$\mathbf{y}=  \mathbf{s} + \mathbf{w},$$
where $\mathbf{s} \in \mathcal{C}$, and $\forall i=1,\ldots, n$, the  $w_i$ are i.i.d. complex Gaussian random variables with variance $\sigma_h^2=\sigma^2=\frac{1}{2}$ per real dimension. (Thus, under the assumptions of the previous Section, the SNR is  $P$). For this channel model we consider the codes
$\mathcal{C}$ of Section \ref{complex_sec}.
Let us  denote with
$$d=\min_{\substack{\mathbf{s}, \bar{\mathbf{s}} \in \mathcal{C}\\ \mathbf{s} \neq \bar{\mathbf{s}}}} \norm{\mathbf{s}-\bar{\mathbf{s}}}$$
the minimum Euclidean distance in the  constellation. Then if ML decoding or naive lattice decoding (NLD)\footnote{By naive lattice decoding, we mean the closest point search in the infinite shifted lattice $x_R+\alpha \psi(\mathcal{O_K})$.} is used, we have the \emph{sphere bound}
$$P_e \leq \mathbb{P}\left\{ \norm{\mathbf{w}}^2 \geq \left(\frac{d}{2}\right)^2\right\}.$$
The minimum distance of the lattice is lower bounded by
\begin{align*}
&d^2 \geq \alpha^2 \min_{x \in \mathcal{O}_K \setminus \{0\}} \norm{ \psi(x)}^2 =\alpha^2 \mathrm{sv}(L)^2=\alpha^2 n.
\end{align*}
Thus, the error probability is bounded by
$$P_e \leq \mathbb{P}\left\{ \norm{\mathbf{w}}^2 \geq \left( \frac{\alpha^2 n}{4}\right)\right\}.$$ 
Note that $2\norm{\mathbf{w}}^2 \sim \chi^2(2n)$. For a random variable $Z \sim \chi^2(n)$, the following concentration result holds $\forall \varepsilon>0$ \cite{Laurent_Massart}:
$$\mathbb{P}\left\{\frac{Z}{n} \geq 1 + \epsilon \right\} \leq 2 e^{-\frac{n\epsilon^2}{16}}.$$  
Consequently, the probability of the set of non-typical noise vectors vanishes exponentially fast: 
$$\mathbb{P}\left\{\frac{\norm{\mathbf{w}}^2}{n} \geq 1 + \epsilon \right\} \leq 2 e^{-\frac{n\epsilon^2}{8}}.$$  
Therefore, $P_e \to 0$ when $n \to \infty$ provided that
$$2^R<\frac{P C_n^{\frac{1}{n}}}{(1+\epsilon)2G} .$$
As $C_n=\frac{(\pi n)^n}{n!}$, using Stirling's approximation we have $C_n \approx \frac{(\pi e)^n}{\sqrt{2\pi n}}$ for large $n$. We can conclude that $P_e \to 0$ for any rate
$$ R < \log_2(P) - \log_2(2G(1+\varepsilon))+\log_2(\pi e).$$
Since the previous bounds hold $\forall \epsilon$, we get the following:
\begin{prop}
Over the complex Gaussian channel, any rate 
$$ R < \log_2(P) - \log_2\left(\frac{2G}{\pi e}\right)$$
is achievable with the code construction in Section \ref{complex_sec}.
\end{prop}

\subsection{Real constellations}

We consider a real Gaussian channel model
$$\mathbf{y}=  \mathbf{s} + \mathbf{w},$$
where $\mathbf{s} \in \mathcal{C}$, and $\forall i=1,\ldots, n$, the  $w_i$ are i.i.d. real Gaussian random variables with variance $\sigma_h^2=\sigma^2=1$. 
The finite codes we consider are those of section \ref{real_sec}.

Analogously to the complex case we have $d^2 \geq\alpha^2 \mathrm{sv}(L)^2=\alpha^2 n$
and
$$P_e \leq \mathbb{P}\left\{ \norm{\mathbf{w}}^2 \geq \left( \frac{\alpha^2 n}{4}\right)\right\}.$$ 

For all $\varepsilon>0$, the error probability vanishes as long as
$$2^{2R}<\frac{P (C^{\R}_n)^{\frac{2}{n}}}{4(1+\epsilon)G_1}. $$

Using Stirling's approximation $C_n^{\R}\approx \frac{(2\pi e)^{n/2}}{\sqrt{\pi n}}$, we get the following:
\begin{prop}
Over the real Gaussian channel, any rate
$$ R < \frac{1}{2}\log_2(P) -\frac{1}{2} \log_2\left(\frac{2G_1}{\pi e}\right)$$
is achievable using the code construction in Section \ref{real_sec}.
\end{prop}

\section{Number field codes in the fast fading channel}%

\subsection{Complex fast Rayleigh fading channel} \label{complex_fast_fading}

We consider a complex fast Rayleigh fading channel model 
$$\mathbf{y}=\mathbf{h} \cdot \mathbf{s} + \mathbf{w},$$
where $\mathbf{s} \in \mathcal{C} \subset \C^n$, and $\forall i=1,\ldots, n$, the $h_i$, $w_i$ are i.i.d. complex Gaussian random variables with variance $\sigma_h^2=\sigma^2=\frac{1}{2}$ per real dimension. Therefore, if $\mathcal{C}$ is one of the lattice codes described in Section \ref{complex_sec}, the SNR is equal to $P$. \\
The minimum distance in the received constellation is
$$d_{\mathbf{h}}=\min_{\substack{\mathbf{s}, \bar{\mathbf{s}} \in \mathcal{C}\\ \mathbf{s} \neq \bar{\mathbf{s}}}} \norm{\mathbf{h} \cdot (\mathbf{s}-\bar{\mathbf{s}})}.$$
The ML and NLD error probabilities are both bounded by
$$P_e \leq \mathbb{P}\left\{ \norm{\mathbf{w}}^2 \geq \left(\frac{d_{\mathbf{h}}}{2}\right)^2\right\}.$$ 
From the arithmetic-geometric mean inequality, we get
\begin{align*}
&d_{\mathbf{h}}^2 \geq \alpha^2 \min_{x \in \mathcal{O}_K \setminus \{0\}} \norm{ \mathbf{h} \cdot \psi(x)}^2=\\
&=\alpha^2 \min _{x \in \mathcal{O}_K \setminus \{0\}} \sum_{i=1}^n \abs{h_i}^2 \abs{\sigma_i(x)}^2 \geq \\
&\geq \alpha^2 \min _{x \in \mathcal{O}_K \setminus \{0\}}  n\left(\prod_{i=1}^n \abs{h_i}^2 \abs{\sigma_i(x)}^2\right)^{\frac{1}{n}}.
\end{align*}
Since $\prod_{i=1}^n \abs{\sigma_i(x)}\geq 1$ for all $x \in \mathcal{O}_K \setminus \{0\}$, we have
$$ d_{\mathbf{h}}^2 \geq \alpha^2 n\left(\prod_{i=1}^n \abs{h_i}^2 \right)^{\frac{1}{n}}$$
Therefore we have the upper bound
\begin{equation} \label{upper_bound}
P_e \leq \mathbb{P}\left\{ \frac{\norm{\mathbf{w}}^2}{n} \geq \frac{\alpha^2}{4} \left(\prod_{i=1}^n \abs{h_i}^2 \right)^{\frac{1}{n}}\right\}.
\end{equation}
Since the $\abs{h_i}$ are Rayleigh distributed with parameter $\sigma_h^2=\frac{1}{2}$, the random variables $X_i=\abs{h_i}^2$ have exponential density $p_{X}(x)=e^{-x}$. 
To find a good upper bound for the error probability, we need to analyze the distribution of the random variable $V_n= \left(\prod_{i=1}^n X_i \right)^{\frac{1}{n}}$, which is a geometric average of exponential distributions. 

Note that $\ln V_n=\frac{1}{n} \sum_{i=1}^n \ln X_i$. The random variables $Y_i=\ln X_i$ have density $p_Y(y)=e^{y-e^{y}}$ and mean
$$ m_y=\mathbb{E}[\ln X]=\int_{0}^{\infty} (\ln x) e^{-x} dx=-\gamma,$$
where $\gamma \approx 0.577215$ is the Euler-Mascheroni constant. From the Chernoff bound \cite[\S 2.1.6]{Proakis} for the zero-mean random variable $-\frac{1}{n}\sum_{i=1}^n \ln X_i -\gamma$, we get that $\forall \delta,v >0, \forall v>0$,
\begin{equation} \label{Chernoff_bound}
\mathbb{P}\left\{\frac{1}{n}\sum_{i=1}^n \ln X_i \leq -(\delta+\gamma)\right\}\leq e^{-nv(\delta+\gamma)} \left(\mathbb{E}[e^{-vX}]\right)^n
\end{equation}
For a given $\delta>0$, the optimal $v_{\delta}>0$ that gives the tightest upper bound is the solution of the equation $\mathbb{E}[-\ln X e^{-v_{\delta} \ln X}]=(\delta + \gamma) \mathbb{E}[e^{-v_{\delta} \ln X}]$. We have
\begin{align*}
&\mathbb{E}[e^{-v\ln X}]=\int_{0}^{\infty} \frac{e^{-x}}{x^v} dx=\Gamma(1-v),\\
&\mathbb{E}[-\ln X e^{-v\ln X}]=\int_{0}^{\infty} \frac{\ln x e^{-x}}{x^v} dx=-\Gamma(1-v)\psi(1-v), 
\end{align*} 
where $\psi(x)=\frac{d}{dx}\ln\Gamma(x)$ denotes the digamma function. Thus, $\psi(1-v_{\delta})=-(\delta+\gamma)$. Note that as $\delta \to 0$, also $v_{\delta} \to 0$ since $\psi(1)=-\gamma$. 
The Chernoff bound (\ref{Chernoff_bound}) thus gives
\begin{align*}
&\mathbb{P}\left\{\ln V_n \leq -(\delta+\gamma)\right\}=\mathbb{P}\{ V_n \leq e^{-\delta} e^{-\gamma}\} \leq \\
& \leq e^{-nv_{\delta}(\gamma + \delta)}(\Gamma(1-v_{\delta}))^n= e^{n(v_{\delta} \psi(1-v_{\delta})+\ln\Gamma(1-v_{\delta}))} 
\end{align*}
The mean value theorem for the function $\ln\Gamma(x)$ in the interval $[1-v_{\delta},1]$ yields $\abs{\ln\Gamma(1-v_{\delta})} \leq \abs{\psi(\xi)} v_{\delta}$
for some $\xi \in (1-v_{\delta},1)$. Since $\psi <0$ in the interval $(0,1)$, $\abs{\psi(\xi)} \leq \abs{\psi(1-v_{\delta})}=-\psi(1-v_{\delta})$, and so 
$$v_{\delta} \psi(1-v_{\delta})+\ln \Gamma (1-v_{\delta}) \leq 0.$$ Therefore $\forall \delta>0$, $\mathbb{P}\left\{\ln V_n \leq -(\delta+\gamma)\right\} \to 0$ as $n \to \infty$.  
\smallskip \par
Fix $\epsilon>0$. Going back to the bound (\ref{upper_bound}), the law of total probability implies that 
{ \allowdisplaybreaks
\begin{align*}
&P_e 
\leq \mathbb{P}\left\{ \frac{\norm{\mathbf{w}}^2}{n} \geq 1+\epsilon \right\} + \mathbb{P}\left\{ \frac{\alpha^2}{4} V_n < 1+ \epsilon\right\}.
\end{align*}
}%
As seen in the Gaussian case, the first term in the previous sum vanishes exponentially fast. 
The second term will tend to $0$ when $n \to \infty$ provided that 
$\frac{4(1+\epsilon)}{\alpha^2} < e^{-(\delta+\gamma)}$
Therefore, $P_e \to 0$ provided that
$$2^R<\frac{P C_n^{\frac{1}{n}}}{2e^{\delta+\gamma}(1+\epsilon)d_K^{\frac{1}{2n}}} =\frac{P C_n^{\frac{1}{n}}}{2e^{\delta+\gamma}(1+\varepsilon)G}.$$
Again using Stirling's approximation we have $C_n \approx \frac{(\pi e)^n}{\sqrt{2\pi n}}$ for large $n$, and the achievable rate is 
$$ R < \log_2(P) - \log_2\left(\frac{2G(1+\varepsilon)e^{\delta+\gamma}}{\pi e}\right)$$
Since the previous bounds hold for any choice of $\epsilon, \delta>0$, we can state the following:
\begin{prop} Over the complex Rayleigh fading channel, any rate 
$$ R < \log_2(Pe^{-\gamma}) - \log_2\left(\frac{2G}{\pi e}\right)$$
is achievable using the codes of Section \ref{complex_sec}. 
\end{prop}
We can compare this result to the bound for Rayleigh channel capacity given in \cite{ONBP_2002}, equation (7):
$$C \geq \log_2(1+P e^{-\gamma}).$$
This is a lower bound, however it has been shown to be very tight for high SNR.

\subsection{Real Rayleigh fast fading channel} \label{real_fast_fading}

We consider a real fast Rayleigh fading channel model \cite{BERB} 
$$\mathbf{y}=\mathbf{g} \cdot \mathbf{s} + \mathbf{w},$$
where $\mathbf{s} \in \mathcal{C}$, and $\forall i=1,\ldots, n$, the $g_i=\abs{h_i}$ are Rayleigh distributed with parameter $\sigma_h^2=\frac{1}{2}$, and $w_i$ are i.i.d. real Gaussian random variables with variance $\sigma^2=1$. Note that the SNR is again $P$ when using one of the real lattice constellations from Section \ref{real_sec}. 
The error probability estimate for this model proceeds exactly as  in the case of the complex Rayleigh fading channel in Section \ref{complex_fast_fading}. 
 A sufficient condition to have vanishing error probability when $n \to \infty$ is
$$2^{2R}<\frac{P (C_n^{\R})^{\frac{1}{n}}}{4e^{\delta+\gamma}(1+\epsilon)d_K^{\frac{1}{2n}}} \approx \frac{P (C_n^{\R})^{\frac{1}{n}}}{4e^{\delta+\gamma}(1+\epsilon)G_1}.$$
Since $C_n^{\R} \approx \frac{(2\pi e)^n}{\sqrt{\pi n}}$ for large $n$, and taking the supremum over all $\epsilon>0$, we find the following:
\begin{prop} Over the real Rayleigh fading channel, any rate
$$ R < \frac{1}{2}\log_2(Pe^{-\gamma}) - \frac{1}{2}\log_2\left(\frac{2G_1}{\pi e}\right)$$
is achievable using the codes of Section \ref{real_sec}.
\end{prop}

\section{Discussion}
Let us now draw some conclusions and highlight the similarities between Gaussian and fast-fading channels.
We saw that there exists an ensemble of lattice codes from number fields that reach all rates satisfying
$$
R < \frac{1}{2}\log_2(Pe^{-\gamma}) - \frac{1}{2}\log_2\left(\frac{2G_1}{\pi e}\right)
$$
in real fast fading channels and rates
$$
R < \frac{1}{2}\log_2(P) - \frac{1}{2}\log_2\left(\frac{2G_1}{\pi e}\right),
$$
in Gaussian channel.
According to \eqref{realeq} these results can be transformed into the following forms
\begin{equation}\label{product}
R < \frac{1}{2}\log_2(Pe^{-\gamma}) - \frac{1}{2}\log_2\left(\frac{2}{\pi e (\mathrm{Nd}_{(\mathrm{p, min})}(L))^{2/n} }\right)
\end{equation}
$$
R < \frac{1}{2}\log_2(P) - \frac{1}{2}\log_2\left(\frac{2n}{\mathrm{Nsv}(L)^{2}\pi e}\right).
$$
Here the normalized product distance and shortest vector  play identical roles. The greater the distance, the smaller the gap to capacity.
This is not only a property of these specific number field codes, but is true for any family of lattice codes. Indeed, while our proofs refer to specific number field codes, the performance only depends on the normalized product distances.

\smallskip
We can now see that in order to reach a constant gap to capacity in fast fading channel, at least with this method, we must
have that $(\mathrm{Nd}_{(\mathrm{p, min})}(L_n))^{2/n}$ stays above some constant. According to Proposition \ref{simple} the product distance is upperbounded by the Hermite constant of the lattice. This result suggests that when $n$ grows a lattice code must have a Hermite constant growing linearly with $n$ in order to be good over the fast fading channel. However, we note that a good Hermite constant does not automatically guarantee a good performance in fast fading channels for general families of lattice codes.

\smallskip
Finally, let us consider how close to capacity this approach can bring us in an optimal scenario. If we consider totally real lattices from number fields, then  
the Odlyzko bound states that when $m\rightarrow\infty$ we have that  $|d_K|^{1/m}\geq 60.8$. Assuming that we can reach this bound with an ensemble of lattice codes we have that any rate $R$ satisfying
$$
R < \frac{1}{2}\log_2(Pe^{-\gamma}) - \frac{1}{2}\log_2\left(\frac{2 \cdot 60.8}{\pi e}\right)
$$
is achievable.
The Odlyzko bound does bound the achievable rate of number field codes, but if we consider all lattices  we have a slightly weaker bound.
For a full lattice in $\R^n$, a classical result of Minkowski gives us that $\mathrm{Nd}_{\mathrm{p, min}}(L)\leq \frac{n!}{n^n}$. Assuming that we have an ensemble of lattice codes reaching this bound 
we have by Stirling's approximation and  equation \eqref{product} that rates satisfying
$$
R < \frac{1}{2}\log_2(Pe^{-\gamma}) - \frac{1}{2}\log_2\left(\frac{2e}{\pi}\right),
$$
are achievable. This result shows that with this method we will always have a gap to $\frac{1}{2}\log_2(Pe^{-\gamma})$ irrespective of the choice of lattice code. However, just like in the case of the Gaussian channel, these bounds do not represent the performance limits of lattice codes, because the method itself and the error probability bounds are suboptimal.

\begin{rem}
We note that the number field towers we used are not the best known possible.  It was shown in \cite {Hajir_Maire} that one can construct a family of real fields such that $G_1<954.3$  and totally complex such that $G<82.2$, but this choice would add some notational complications.
\end{rem}

\bibliographystyle{IEEEtran}
\newpage
\section{Appendix: Increasing the product distance using ideals}
In the previous sections we were considering just the ring of algebraic integers $\OO_K$ and the corresponding lattice $\psi(\OO_K)$. 
Just as well we could have considered any other  additively closed subgroup of $\OO_K$ and in particular  ideals of $\OO_K$.  
 In most works on number field lattices the authors  were concentrating on either the ring $\OO_K$ or a principal ideal  $a\OO_K$. In   \cite{FOV} the authors were also considering the question of increasing the normalized product distance and achievable rate by using a non-principal ideals $I$.

While 
finding the normalized product distance of  lattices $\psi(\OO_K)$ or  $\psi(a\OO_K)$ is an easy task, the same is not true for a non principal ideal $I$. In this appendix we will show how this problem can be reduced to another more well known problem in algebraic number theory and how it can be used to study the  performance limits of lattices $\psi(I)$.

\subsection{Ideals in totally complex fields}

Let us suppose that $K$ is  degree $2n$ totally complex field. We will use the notation $\mathrm{N}(I)=[\OO_K:I]$, for the norm of an ideal $I$ and $nr_{K/\Q}(x)$ for the norm of an element $x$ in $K$. From classical algebraic number theory we  have that $N(a\OO_K)=|nr_{K/\Q}(a)|$ and 
$N(AB)=N(A)N(B)$.

\begin{lem}\label{idealvolume}
Let us suppose that $K$ is a totally complex field of degree $2n$ and that  $I$ is an integral ideal in $K$. We then have that $\psi(I)$ is
a $2n$-dimensional lattice in $\C^n$ and that
$$
\Vol{(\psi(I))}=[\OO_K:I] 2^{-n}\sqrt{|d_K|}.
$$
\end{lem}
This well known result gives the volume of an ideal, but the question of the size of the normalized product distance of an ideal is a more complicated issue. In \cite[Theorem 3.1]{FOV} the authors stated the analogue of the following  result for the totally real case. It is simply a restatement of the definitions.
\begin{prop} 
Let us suppose that $K$ is a totally complex field of degree $2n$ and that $I$ is an integral ideal of $K$. We then have that
\begin{equation}
\mathrm{Nd_{p,min}}(\psi(I))=\frac{2^{\frac{n}{2}}}{|d_K|^{\frac{1}{4}}}\mathrm{min}(I),
\end{equation}
where $\mathrm{min}(I):=\underset{x \in I \setminus \{0\}}{\mathrm{min}}\sqrt{\frac{|\mathrm{nr}_{K/\Q}(x)|}{\mathrm{N}(I)}}.$
\end{prop}
\begin{proof}
This  result follows from Lemma \ref{idealvolume}, the definition of the normalized product distance and from noticing that
 $\sqrt{|\mathrm{nr}_{K/\Q}(x)|}=|n(\psi(x))|.$
\end{proof}

Due to the basic ideal theory of algebraic numbers $\mathrm{min}(I)$ is always larger or equal to $1$. If $I$ is not a principal ideal then we have that
$\mathrm{min}(I)\geq \sqrt{2}$. Comparing this to Proposition \ref{complex} we find that, given a non principal ideal domain $\OO_K$, we  should use an ideal $I$, which is not principal, to maximize the product distance. Now there are two obvious questions. Given a  non principal ideal domain $\OO_K$, which ideal $I$ we should use and how much we gain if the used ideal is optimal? 
Before answering these questions we need the following.

\begin{lem}\cite{FOV}
Let us suppose that $x$ is any element from $K$. We then have that
$$
\mathrm{Nd_{p,min}}(\psi(xI))=\mathrm{Nd_{p,min}}(\psi(I)).
$$
\end{lem}
This result proves that every ideal in a given ideal class has the same normalized product distance. It follows that given a ring of integers $\OO_K$, it is enough to check one ideal from every ideal class to find the optimal ideal. Given an ideal $I$ we will denote with $[I]$ the ideal class where ideal $I$ belongs.

Let us denote with $N_{min}(K)$  the norm of an ideal $A$ in $K$ with the property that every ideal class of $K$ contains an integral ideal with norm $N(A)$ or smaller. 

\begin{prop}\label{idealform}
Let us suppose that $K$ is a totally complex number field  and that $I$ is an ideal that maximizes the normalized product distance over  all
ideals in $K$. We then have that
$$
\mathrm{Nd_{p, min}}(\psi(I))= \frac{2^{n/2}\sqrt{N_{min}(K)}}{|d_K|^{\frac{1}{4}}}.
$$
\end{prop}
\begin{proof}
Let us suppose that $L$ is any 
ideal in $K$. Let us also suppose that $A$ is an integral ideal in class  $[L]^{-1}$ with the smallest norm. We then have that there exists an element $y \in \OO_K$ such that $y\OO_K=AL$. As $n(\psi(y))=\sqrt{N(L)N(A)}$ and $N(A)\leq N_{min}(K)$ we have that
$\mathrm{d_{p,min}}(L) \leq \sqrt{N(L)N_{min}(K)}$ and $\mathrm{Nd_{p, min}}(L)\leq  \frac{\sqrt{N_{min}(K)}2^{n/2}}{|d_K|^{1/4}}$. 

Let us assume that $S$ is such an ideal that $N(S)=N_{min}(K)$ and choose $I$ as an element from class
$[S]^{-1}$. Let us now suppose that $x$ is any non-zero element of $I$. We then have that $x\OO_K=IC$, for some ideal $C$ that belongs to class
$[S]$. Therefore we have that $n(\psi(x))\geq \sqrt{N(I)N(C)}$.
\end{proof}
This result translates the question of product distance of an ideal to well known problem in algebraic number theory. It does also describes which ideal class we should use in order to maximize the product distance.

Let us denote with $\mathcal{K}_{2n}$ the set of totally complex number fields  of degree $2n$. We then have that the 
optimal normalized product distance over all degree $2n$-complex fields and all ideals $I$ is
$$
\underset{K\in \mathcal{K}_{2n }}{\mathrm{min}} \frac{2^{n/2}\sqrt{N_{min}(K)}}{|d_K|^{\frac{1}{4}}}.
$$

The following theorem by Zimmert \cite{Zim} then gives us an upper bound of what can be achieved with this method.

\begin{theorem}
Assuming that we have a number field $K$ with    signature $(r_1,r_2)$ we have that
$$
N_{min}(K)\leq ((50.7)^{r_1/2} (19.9)^{r_2})^{-1}\sqrt{|d_K|},
$$
when  $[K:\Q]$ is large enough.
\end{theorem}

\begin{cor}\label{idealbound}
Given a totally complex field $K$ of  degree $2n$  and any ideal $I \subset K$ we have that
$$
\mathrm{Nd_{p, min}}(I)\leq (3.1)^{-n},
$$
when $n$ is large enough.
\end{cor}

\subsection{Ideals in totally real fields}
Let us now state the analogous results for the totally real case. 

\begin{lem}
Let us suppose that $K$ is a totally real field of degree $n$ and that  $I$ is an integral ideal in $K$. We then have that $\psi(I)$ is
a $n$-dimensional lattice in $\R^n$ and that
$$
\Vol{(\psi(I))}=[\OO_K:I]\sqrt{|d_K|}.
$$
\end{lem}

\begin{prop}\cite{FOV}
Let us suppose that $K$ is a totally real degree $n$ number field and that $I$ is an integral ideal of $K$. We then have that
\begin{equation}
\mathrm{Nd_{p,min}}(\psi(I))=\frac{1}{|d_K|^{\frac{1}{2}}}\mathrm{min}(I),
\end{equation}
where $\mathrm{min}(I):=\underset{x\neq 0\in I}{\mathrm{min}}\frac{|\mathrm{nr}_{K/\Q}(x)|}{\mathrm{N}(I)}.$
\end{prop}

\begin{prop} \label{real_idealform} 
Let us suppose that $K$ is a totally real number field  and that $I$ is such an ideal that it maximizes the normalized product distance over  all
ideals in $K$. We then have that
$$
\mathrm{Nd_{p, min}}(\psi(I))= \frac{N_{min}(K)}{|d_K|^{\frac{1}{2}}}.
$$
\end{prop}

Let us denote with $\mathcal{K}_{n}$ the set of totally real  number fields  of degree $n$. We then have that the 
optimal normalized product distance over all degree $n$ real fields and all ideals $I$ is
$$
\underset{K\in \mathcal{K}_{n }}{\mathrm{min}} \frac{N_{min}(K)}{|d_K|^{\frac{1}{2}}}.
$$

\begin{cor}\label{realideal}
Given a totally real number field $K$ of degree $n$  and any ideal $I$ we have that
$$
\mathrm{Nd_{p, min}}(I)\leq (7.12)^{-n},
$$
when $n$  is large enough.
\end{cor}

\subsection{Final remarks}

\begin{rem}
The relation in Propositions \ref{idealform} and \ref{real_idealform} can be used in the opposite direction to  derive bounds for the value of $N_{min}(K)$ from product distance bounds. Therefore Corollaries \ref{idealbound}  and \ref{realideal}
are just  better versions of the Minkowski bound already given in \cite[Section 2.4]{Oggier}.  However, to state our corollaries one has to go through an argument similar to Propositions \ref{idealform} and \ref{real_idealform}, as  the theorem of Zimmert used ideal-theoretic and analytic methods that are  not directly applicable to the normalized minimum determinant problem. This is in contrast to the Minkowski bound, which is completely general and is based on geometry of numbers.

Moreover, the formulation given in Propositions \ref{idealform} and \ref{real_idealform} can be very beneficial when studying what can be achieved using non-principal ideals. For example when the class number is $2$, the value of $N_{min}(K)$ is simply the smallest norm among the non-principal ideals in $K$. Just as well this result describes which ideal class we should use.
\end{rem}

\begin{rem}
As was already pointed out in \cite{FOV} and \cite[p. 52]{Oggier} there is no guarantee that we can really gain something by using non-principal ideals. While it is true that in number fields with class number greater than one, using a non-principal ideal does give us some gain, this gain may not be enough to compensate for the possibly large discriminant of these number fields. This trade-off is clear from Proposition \ref{idealform}.

\end{rem}

\begin{rem}
We note that this ideal approach can  already be used to increase achievable rates of the number field constructions in Sections \ref{complex_fast_fading} and \ref{real_fast_fading}.
This is due to the fact that all the fields in the Martinet families have class number larger than 1 and therefore the corresponding rings of algebraic integers are not principal ideal domains.
\end{rem}

\bibliographystyle{IEEEtran}
{\small

\end{document}